\pdfoutput=1
\documentclass[a4paper,11pt]{llncs}
\usepackage{tabularx}
\usepackage{amsmath}
\usepackage{amssymb}
\usepackage{algorithm}
\usepackage{algorithmic}

\usepackage{color}
\usepackage{graphicx}
\usepackage{amsfonts}
\usepackage[in]{fullpage}

\title{Streaming Algorithms for Partitioning Integer Sequences}

\author{Christian Konrad \inst{1} \and L\'aszl\'o Kozma \inst{2}}

\institute{Reykjavik University, Reykjavik, Iceland \email{christiank@ru.is} \and Universit\"{a}t des Saarlandes, Saarbr\"{u}cken, Germany \email{kozma@cs.uni-saarland.de}}

\newcommand{\polylog}{\mathop{\mathrm{polylog}}}

\def\lc{\left\lceil}   
\def\rc{\right\rceil}
\def\lf{\left\lfloor}   
\def\rf{\right\rfloor}

\newcommand{\Order}{\mathrm{O}}
\newcommand{\Exp}{\mathrm{Exp}}

\newcommand{\comment}[1]{\textcolor{blue}{\textbf{#1}}}
\newcommand{\lk}[1]{\comment{\textcolor{blue}{#1}}}

\newcommand{\ignore}[1]{}

\begin{document}

\maketitle

\begin{abstract}
 We study the problem of partitioning integer sequences in the one-pass data streaming model. Given is an input stream
of integers $X \in \{0, 1, \dots, m \}^n$ of length $n$ with maximum element $m$, and a parameter $p$. The goal is to 
output the positions of separators splitting the input stream into $p$ contiguous blocks such that the maximal weight
of a block is minimized. We show that computing an optimal solution requires linear space, 
and we design space efficient $(1+\epsilon)$-approximation algorithms for this problem following the parametric search framework. 
We demonstrate that parametric search can be successfully applied in the streaming model, and we present
more space efficient refinements of the basic method. All discussed algorithms require space
$\Order( \frac{1}{\epsilon} \polylog (m,n,\frac{1}{\epsilon}))$, and we prove that the linear dependency on $\frac{1}{\epsilon}$
is necessary for any possibly randomized one-pass streaming algorithm that computes a $(1+\epsilon)$-approximation.
\end{abstract}

% \newpage

\section{Introduction}
In this paper, we study the problem of \textit{partitioning integer sequences}. Given a sequence of integers $X \in \{0, 1, \dots, m\}^n$ of length $n$, with maximum element $m$, and an integer $p \ge 2$, the goal is to partition $X$ into $p$ contiguous blocks
such that the maximum weight (sum of the elements) of a block is minimized. In other words, we have to find $p-1$ 
separators $s_1, \dots, s_{p-1}$ with $1 = s_0 \le s_1 \le \dots \le s_{p-1} \le s_p = n+1$ such that 
$$\max \left\{ \sum_{i=s_j}^{s_{j+1}-1}X_i \, \ \Bigl\vert\Bigr. \ \, j \in \{0, \dots, p-1\} \right\}$$ is minimized. The
value of the previous expression is called the \textit{bottleneck value} of the partitioning. In the following, for any integer $j \in \{0, \dots, p-1\}$ we refer to the elements $\{X_{s_j}, \dots, X_{s_{j+1}-1} \}$ as a \emph{partition}, and we refer to the sum of these elements as the \emph{weight} of the partition.

This problem appears in many applications, especially in the context of load balancing, and has been extensively 
studied both from a theoretical \cite{b88,hl92,mo95,ms95,kms97,hcn92} and a practical perspective \cite{mp97,pa04}.
In the literature, it appears under various names such as \textit{chains-on-chains partitioning}
\cite{pa04,hcn92} or \textit{1D rectilinear partitioning} \cite{mp97}.

Very efficient exact algorithms for this problem exist, for example the $\Order(n \log n)$ time algorithm of Khanna et al.\ \cite{kms97}, 
the $\Order(n + p^{1+\epsilon})$ time algorithm of Han et al.\ \cite{hcn92}, and the optimal $\Order(n)$ time algorithm
of Frederickson \cite{f91}. However, all existing approaches require either 
random access to the input or at least multiple access to the same input element. Since in many applications
the input integer sequences are huge and cannot be entirely stored in a 
computer's random access memory, data access is a bottleneck for the previously mentioned algorithms.  
One example application is the decomposition of computational meshes along space filling curves \cite{sw05,k11,b13}. In
parallel scientific computing, for instance in the area of parallel particle simulations or parallel solutions of partial differential 
equations, huge meshes have to be decomposed and distributed to different computational units.
In the space filling curves approach, mesh elements are linearly ordered along a space filling curve
which allows the reduction of the multi-dimensional decomposition problem to the one-dimensional problem of partitioning 
integer sequences, the problem studied in this paper. Today, meshes of Gigabyte or 
even Terabyte size are common and exceed by far a computer's random access memory. Algorithms for this problem 
should therefore have an IO-efficient memory access pattern. In this paper, we are therefore interested in 
\textit{streaming algorithms} for the problem of partitioning integer sequences.

\textit{Streaming Model.} In the data streaming model, an algorithm receives its input as a data stream piece by piece. 
The algorithm is granted
a small random access memory which is often only polylogarithmic in the input size. In the present work, we focus on one-pass streaming algorithms, however, depending on the application, an algorithm may be granted multiple passes over 
the input data in order to further decrease the size of its random access memory. Streaming algorithms find 
applications in situations where the input data is too large to be stored in local memory and random data access is 
too costly. %They fall therefore into the category of massive data set algorithms. 
For an introduction to streaming 
algorithms, we refer the reader to \cite{mut05}.

\textit{Streaming Algorithms for Partitioning Integer Sequences.}
We assume that our streaming algorithms receive an input stream $X \in \{0, 1, \dots, m\}^n$ of 
length $n$ consisting of integers from the set $\{0, 1, \dots, m\}$. In addition, we assume that the number of partitions
to be created $p$ is stored in the random access memory. All our algorithms make a single pass over the 
input stream. Since we show that any streaming algorithm that computes an exact solution requires $\Omega(n)$ space, 
we consider approximation algorithms. We say that an algorithm is a \textit{$c$-approximation algorithm} if
it computes a partitioning with a bottleneck value which is larger than the optimal bottleneck
value by at most a factor $c$. All our algorithms are deterministic. Nevertheless, we prove space lower bounds for possibly \textit{randomized} algorithms. 
A randomized streaming algorithm is a streaming algorithm that has access to an infinite sequence of random bits, and 
outputs a correct solution with probability at least $1-\delta$, for a small constant $\delta$. 

We consider the following two variants of the problem:
\begin{enumerate}
 \item The streaming algorithm outputs separators $s_0, s_1, \dots, s_p$ that determine the positions of the partitions
in the stream. We abbreviate this variant of the problem by \textsc{Part}.
 \item The streaming algorithm outputs an upper bound on the bottleneck value of an optimal partitioning. We abbreviate this
variant of the problem by \textsc{PartB} (\textsc{B} stands for bottleneck).
\end{enumerate}
There is an important relation between the two variants \textsc{Part} and \textsc{PartB}. A solution to \textsc{PartB}, i.e.,
a bottleneck value, can be transformed into a solution to \textsc{Part}, i.e., the partition boundaries, via one
additional pass over the input stream using the \textsc{Probe} algorithm which is used in many prior works
on this problem, e.g. \cite{i91,k11,kms97}.
\textsc{Probe} takes a bottleneck value $B$ and 
traverses the stream $X$ creating maximal partitions of weight at most $B$. It is easy to see that \textsc{Probe} succeeds if 
and only if $B$ is at least as large as the optimal bottleneck value. For this reason, in the definition of \textsc{PartB}, 
we do not allow a streaming algorithm to output a value that is smaller than the optimal bottleneck value. 
The \textsc{Probe} algorithm is also an important building block in our work, and we discuss it in more detail in 
Section~\ref{sec:probe}.

\textit{Parametric Search Algorithms.} The previously described relation between \textsc{Part} and \textsc{PartB} via
the \textsc{Probe} algorithm suggests the application of the \textit{parametric search} framework to this problem, and, in fact, an
optimal $\Order(n)$ time algorithm for this problem is obtained by Frederickson in \cite{f91} via this approach. 
Parametric search was developed by Megiddo more than $30$ years ago \cite{m78,m83} and has become a standard technique. A parametric search 
problem is one where the optimal solution is the smallest (or largest) value from a set of candidate solutions
of an interval $\{a, a+1, \dots, b\}$ that passes a certain feasibility test. Usually, monotonicity holds for the values in 
$\{a, a+1, \dots, b\}$, i.e., if a value $x \in \{a, a+1, \dots, b\}$ is feasible then all values $\{x, \dots, b\}$ 
(respectively $\{a, \dots, x\}$) are also feasible. In this situation, using binary search, an $\Order(\log (b-a) F)$
algorithm can therefore be obtained immediately, where $F$ is the runtime of the feasibility test.

Applied to the problem of partitioning integer sequences, testing feasibility of a value $B$ corresponds to 
a run of the \textsc{Probe} algorithm. A trivial range for the possible bottleneck values is $\{1, \dots, nm \}$
(we discuss better ranges in Section~\ref{sec:probe}), and, therefore,
an $\Order( \log (mn) n )$ time exact algorithm can be obtained. In \cite{f91}, Frederickson improves this basic idea 
and obtains an $\Order(n)$ time
algorithm by building data structures on the input sequence that allow the speeding up of the feasibility test, and by exploiting 
additional information obtained during the feasibility test in order to further narrow down the search space. 

Parametric search strongly relies on the fact that the choice of parameter for the next feasibility test 
depends on the outcome of previous feasibility tests. However, this is impossible to establish in the one-pass
streaming model, and, in fact, we prove that in one pass and sublinear space it is impossible to compute
the optimal bottleneck value. When relaxing to a $(1+\epsilon)$-approximation, the parametric search framework
allows a strategy that results in a one-pass streaming algorithm with space $\Order( \frac{1}{\epsilon} \log (b - a) S)$, where
$S$ is the space required to perform one feasibility test. We run $\Theta(\frac{ \log (b-a)}{\epsilon})$ feasibility tests 
in parallel, testing the values $(1+\epsilon)^{i} a$ for $i \in \{0, \dots, \frac{1}{\epsilon} \lc \log (b - a) \rc \}$, and
we output the smallest parameter of a successful feasibility test. If $m$, the largest element of the stream, 
and $n$, the length of the stream, are known in 
advance to our algorithm, then in one pass a $(1+\epsilon)$-approximation with space 
$\Order (\frac{1}{\epsilon} \log(\frac{mn}{p}) p\log(mn) )$ can be obtained. Note that this algorithm
requires knowledge of the parameters $m$ and $n$ in advance in order to establish a search space that contains the 
optimal bottleneck value. We regard this algorithm as a baseline algorithm
to which we compare our results, and we discuss it in detail in Section~\ref{sec:probe}. 
 
The main contribution of this paper is the design of a new feasibility test: We design the algorithm
\textsc{ProbeExt} that takes a parameter $B$ and outputs a feasible value $2^i B$ that is at most by a factor 
$2+\epsilon$ larger than the optimal bottleneck value, for small $\epsilon$ values, if the optimal bottleneck value $B^*$
is at least $m/ \epsilon^2$. In some sense, if $B^*$ is sufficiently large, this allows us to run $\Theta(\log (b-a))$ 
feasibility tests simultaneously. Therefore, compared to the previously described method of running
$\Theta(\frac{1}{\epsilon} \log (b-a))$ feasibility tests simultaneously, it is enough to 
run only $\Theta( \frac{1}{\epsilon})$ of our improved feasibility tests, which improves the space complexity by a $\log(b-a)$ factor. 
In order to perform our improved feasibility test, we only require knowledge of $m$ in advance while $n$ may be unknown. 
This is somewhat surprising, since the knowledge of $m$ alone does not allow us to determine an upper limit of the 
search space for the optimal bottleneck value. Our improved feasibility test, however, can recover from a failed
test for $x$, and continue running a test for some $y>x$, without having to restart the stream. 
In order to rule out optimal bottleneck values $B^*$ smaller than $m/\epsilon^2$, we additionally run the 
previously discussed \textsc{Probe}
algorithm for bottleneck values in the range $\{m, \dots, m/ \epsilon^2\}$.
This allows us to obtain an $\Order(\frac{1}{\epsilon} \log( \frac{1}{\epsilon}) S)$ space algorithm, where $S$ is the space for 
the \textsc{ProbeExt} algorithm, and the $\log (\frac{1}{\epsilon})$ factor is necessary to rule out cases in which 
the optimal bottleneck value is smaller than $m/\epsilon^2$. 

%\textit{A Baseline Algorithm.} Suppose that the total weight $S = \sum_i X_i$ of the stream is given in advance. Then, running 
%$\log({p}) / \epsilon$ copies of the \textsc{Probe} algorithm allows us to obtain a 
%$(1+\epsilon)$-approximation (see Section~\ref{sec:probe}). We consider this to be a baseline algorithm against which we compare our results. 
%This method may, however, not be applicable in practice, since 
% First,
%its space requirement is $\Omega( \log^2 n )$ in the worst case, both for \textsc{Part} and \textsc{PartB}, which might be prohibitively large, depending on the application. 
%Depending on the application, $m$ may be much larger than $n$, and therefore it is desirable to 
%obtain a linear dependence on $\log m$. 
%Second, the algorithm 
%it requires the knowledge of $S$ in advance.
%(note that it does not require the knowledge of $n$). 
%Again,
%Depending on the application, this may be a reasonable assumption, however, for instance in applications
%where the weights of elements are estimated on-the-fly,
%$S$ is certainly not known. \lk{*} In this work we design a $(1+\epsilon)$-approximation algorithm
%that resolves this issue. 

\textit{Which parameters are known in advance?} The difficulties of \textsc{Part} and \textsc{PartB} depend strongly on
which parameters are known in advance to the algorithm. Suppose that the total weight $S = \sum_i X_i$ of the stream
is known in advance. Then it is easy to argue that the optimal
bottleneck value $B^*$ is such that $S/p \le B^* \le S$. This narrows down the search space, and running
$\Theta(\frac{p}{\epsilon})$ copies of the \textsc{Probe} algorithm is enough to obtain a $(1+\epsilon)$-approximation.
The knowledge of $S$ provides a lot of information about the input stream. Depending on the application, 
this may be a reasonable assumption, however, for instance in applications where the weights of elements are estimated 
on-the-fly, $S$ is certainly not known.
Our $(1+\epsilon)$-approximation algorithm that applies our improved parametric search strategy requires only knowledge of
$m$ in advance (in fact, any value $x$ with $m \le x \le B^*$ will do), while $n$ and $S$ may be unknown. We point out
that, in this situation, the initial search space for bottleneck values is unknown since the length of the stream is not known to 
the algorithm. 
For the situation where no information about the parameters is granted in advance, we are only able to obtain 
a $2$-approximation. We leave the existence of a $(1+\epsilon)$-approximation for this situation as an open question.

\textit{Communication Complexity.} In this paper, we prove two space lower bounds for one-pass streaming algorithms.
We show that computing an optimal solution requires $\Omega(n)$ space, and we show that
computing a $(1+\epsilon)$-approximation requires $\Omega( \frac{1}{\epsilon} \log n )$ space (for any $\epsilon = \Order(n^{1-\gamma})$
for any $\gamma > 0$), 
showing that the $\frac{1}{\epsilon}$ factor is necessary for obtaining a $(1+\epsilon)$-approximation. Proving space 
lower bounds for streaming algorithms is often done via communication 
complexity, and we follow this route in this paper. A one-way two-party communication problem consists of two players, 
usually denoted by Alice and Bob, who
hold inputs $Y$ and $Z$, respectively. Alice sends a single message to Bob who, upon reception, computes the output
of the protocol as a function of Alice's message and his input. The relation to streaming algorithms is as follows: 
A streaming algorithm for a problem $P$ on data stream $X = Y \circ Z$ ($Y$ concatenated with $Z$) with space $s$ can be used 
as a one-way two-party 
communication protocol for problem $P$ with maximal message size $s$ where player one holds input 
$Y$ and player two holds input $Z$. Conversely, a lower bound on the one-way two-party communication complexity of 
a problem $P$ is also a lower bound on the space requirements for any streaming algorithm for problem $P$.
For an introduction to communication complexity, we refer the reader to \cite{kn97}.

\textit{Summary Of Our Results.} Our first result is an impossibility result. We show that 
computing an exact solution to either \textsc{Part} or \textsc{PartB} in one pass requires $\Omega(n)$ space even for randomized algorithms. 
We therefore study approximation algorithms for the problem. 
We show that if the maximal value $m$ of the stream $X$ is known in advance, then
there is a deterministic $(1+\epsilon)$-approximation algorithm for both \textsc{Part} and \textsc{PartB} 
using space 
%$\Order \left( \frac{1}{\epsilon} \log \frac{1}{\epsilon} \left(p \log n + \log \frac{m}{\epsilon}  \right) \right)$ 
$\Order \left( \frac{1}{\epsilon} \log (\frac{1}{\epsilon}) \log ({mn^p})  \right)$ 
%and $\Order \left( \frac{1}{\epsilon} \log \frac{1}{\epsilon} \left(\log p + \log \frac{m}{\epsilon}  \right) + \frac{1}{\epsilon}  \log n \right)$
and $\Order \left( \frac{1}{\epsilon} \log (\frac{1}{\epsilon}) \log \frac{mn}{\epsilon} \right)$,
respectively. These algorithms do not require
knowledge of $n$ or of the total weight $S$ of the stream in advance.
%We show that in the special case of $m=1$ meaning that $X$ is a sequence of $0$s and $1$s, there is an algorithm
%with space $\Order(\frac{\log(n)}{\epsilon})$ (\textbf{Theorem~\ref{}}). This algorithm is compromised with an almost matching 
%lower bound of $\Omega( \frac{\log (n \epsilon)}{\epsilon})$ (\textbf{Theorem~\ref{}}). Note that for 
%$\frac{1}{\epsilon} = \Theta(n^{1-\delta})$ for any $\delta$, the lower bound is tight. In the case of 
%$\frac{1}{\epsilon} = \Theta(n)$, the upper bound stores essentially the entire stream.
%We use the algorithm for the special case of $m=1$ in order to obtain a $(2+\epsilon)$-approximation
%for the case that $m$ is not known to our algorithm (\textbf{Theorem~\ref{}}).
Then, we consider the hardest case when the algorithm has no information about $m, n$ or $S$. 
%We point out that there 
%is a simple $2$-approximation algorithm for \textsc{PartB} using 
%$\Order(\log (m n))$ space. 
We design a $2$-approximation algorithm for \textsc{Part} using space $\Order(p \log (mn) )$,
%$\Order \left( \frac{1}{\epsilon} \log( \frac{1}{\epsilon}) \cdot \left(p \log (nm) + \log \frac{1}{\epsilon}  \right) \right)$. 
%%$\Order \left( \frac{1}{\epsilon} \log (\frac{1}{\epsilon}) \log ({mn^p})  \right) $, 
and point out a simple $2$-approximation algorithm for \textsc{PartB} using space $\Order(\log (m n))$. 
As a counterpoint to these upper bounds, we show that any possibly randomized streaming algorithm
that computes a $(1+\epsilon)$-approximation to \textsc{Part} requires $\Omega( \frac{1}{\epsilon} \log n)$ space
for any $\epsilon = \Order(n^{1-\gamma})$ and any $\gamma > 0$.
As our algorithms have a $\frac{1}{\epsilon} \log \frac{1}{\epsilon}$ dependence on $\epsilon$, our lower bound shows 
that this dependence is optimal up to a logarithmic factor on $\frac{1}{\epsilon}$.
Our results are summarized in Figure~\ref{fig:results}.

\begin{figure}[ht]
\small
\begin{center}
 \begin{tabularx}{\textwidth}{|ccccccX|}
\hline
$\, \, \, \, m \, \, \, \,$ & $\, \, \, \, n \, \, \, \,$ & $\, \, \, \, S \, \, \, \,$ & $\, \, $ Approximation $\, \, $ & 
Space & $\quad \quad$ & Remark \\
\hline 
\multicolumn{7}{|l|}{\textsc{Part}:} \\
 &  &   & exact & $\Omega(n)$ & & Lower bound (Theorem~\ref{thm:lb-exact}) \\
- & - & ! & $1+\epsilon$ & $\Order \left( \frac{1}{\epsilon}\log(p) \log (mn^p) \right)$ & & Baseline (Theorem~\ref{thm:probe}) \\
! & ! & - & $1+\epsilon$ & $\Order \left( \frac{1}{\epsilon} (p \log^2(n) + \log^2(m)) \right)$ & & Baseline (Theorem~\ref{thm:probe2}) \\

%$\Order \left(\frac{1}{\epsilon} \log(pm) (p \log n + \log m)  \right)$ & Baseline (Theorem~\ref{thm:probe}) \\

!  & - & - & $1+\epsilon$ & $\Order \left( \frac{1}{\epsilon} \log (\frac{1}{\epsilon}) \log ({mn^p})  \right)$ & & (Theorem~\ref{thm:known-m})\\
%$\Order \left( \frac{1}{\epsilon} \log \frac{1}{\epsilon} \left(p \log n + \log \frac{m}{\epsilon}  \right) \right)$ & (Theorem~\ref{thm:known-m}) \\
- & - & - & $2$ & $\Order(p \log (mn) )$ & & (Theorem~\ref{thm:unknown-m-part}) \\

%$\Order \left( \frac{1}{\epsilon} \log \frac{1}{\epsilon} \left(p \log (nm) + \log \frac{1}{\epsilon}  \right) \right)$ & (Theorem~\ref{thm:unknown-m}) \\
  & &   & $1+\epsilon$ & $\Omega(\frac{1}{\epsilon} \log n )$ & & Lower bound (Theorem~\ref{thm:lb-approx}) \\
\multicolumn{7}{|l|}{$ $} \\
\multicolumn{7}{|l|}{\textsc{PartB}:} \\
  & &   & exact & $\Omega(n)$ & & Lower bound (Theorem~\ref{thm:lb-exact}) \\
- & -  & !  & $1+\epsilon$ & $\Order \left( \frac{1}{\epsilon}\log(p) \log (mn) \right)$ & & Baseline (Theorem~\ref{thm:probe}) \\
! & !  & -  & $1+\epsilon$ & $\Order \left( \frac{1}{\epsilon} \log^2(mn) \right)$ & & Baseline (Theorem~\ref{thm:probe2}) \\
%$\Order \left( \frac{1}{\epsilon} \log(pm) \log (mn) \right)$
!  & - & - & $1+\epsilon$ & $\Order \left( \frac{1}{\epsilon} \log (\frac{1}{\epsilon}) \log (\frac{mn}{\epsilon}) \right)$ & & (Theorem~\ref{thm:known-m}) \\

%\log \frac{mp}{\epsilon} 

%+ \frac{1}{\epsilon}  

%$\Order \left( \frac{1}{\epsilon} \log \frac{1}{\epsilon} \left(\log p + \log \frac{m}{\epsilon}  \right) + \frac{1}{\epsilon}  \log n \right)$
- & - & - & $2$ & $\Order( \log(mn) )$ & & (Theorem~\ref{thm:unknown-m-partb}) \\
\hline
\end{tabularx}
\caption{Overview of our results. In the first three columns we indicate whether advance knowledge of the maximum weight
$m$, the length of the stream $n$, or the total weight $S$ is required by the algorithm  (the ! sign indicates that the respective quantity is
required).
\label{fig:results}}
\end{center}
\end{figure}

% \begin{table}
% \begin{center}
%   \begin{tabular}{|l c c|}
% \hline
%    & $\quad$  space for \textsc{Part} $\quad$  & $\quad$  space for \textsc{PartB} $\quad$ \\
% \hline
% Upper Bounds: &  & \\
% $m=1$ & $\Order(\frac{\log n}{\epsilon})$ & $\Order \left(\log n (p + \frac{1}{\epsilon})\right)$ \\
% arbitrary but known $m$ & $\Order(\frac{\log n}{\epsilon})$ & $\Order \left(\log n (p + \frac{1}{\epsilon})\right)$ \\
% unknown $m$ & & \\
% & & \\
% Lower Bounds: &  & \\
% exact Algorithms & $\Omega(n)$ & $\Omega(n)$ \\
% $(1+\epsilon)$-approximation & $\Omega( \frac{\log (\epsilon n)}{\epsilon} )$ & - \\
% \hline
%  \end{tabular}
% 
% \caption{bla}
% \end{center}
% 
% \end{table}

%Then, we show in Theorem~\ref{thm:lb-approx} that randomized streaming algorithms that compute a $(1+\epsilon)$-approximation 
%to \textsc{Part} require $\Omega(\frac{ \log (\epsilon n) }{\epsilon})$ space, and we will argue that this lower bound is 
%essentially optimal.

\vspace{-0.9cm}

\textit{Further Related Work.} The problem of partitioning integer sequences has been extensively studied in the offline setting, 
as early as 1988 by Bokhari \cite{b88}, who presented an exact algorithm
with time complexity $\Order(n^3 p)$.  Significant progress has since been made on the problem, and the best current algorithm
runs in time $\Order(n)$ independently of $p$ \cite{f91}.
Previous works use techniques such as dynamic programming, iterative refinement of a partitioning, and parametric 
search. Most ideas from previous works are not applicable in the streaming model
since they require a more flexible data access scheme. The work of Iqbal \cite{i91} is closest to our work because it 
considers approximation algorithms. Furthermore, some of his techniques, such as a parametric search for the optimal bottleneck 
value, are in their basic features similar to our work. %We will discuss these similarities when we discuss our techniques.
To the best of our knowledge, our work is the first that rigorously follows the parametric search framework in the streaming model.

\textit{Outline.} First,
we discuss the \textsc{Probe} algorithm and we prove the results for our baseline method in Section~\ref{sec:probe}. 
In Section~\ref{sec:probe-ext}, we discuss the \textsc{ProbeExt} algorithm, which constitutes the main algorithmic 
technique in this paper. 
%This algorithm is then used in Section~\ref{sec:known-m} and Section~\ref{sec:unknown-m}. 
In Section~\ref{sec:known-m}, we present algorithms for the case when $m$ is known in advance, and in 
Section~\ref{sec:unknown-m}, we present algorithms for the case when $m$ is not known in advance. 
Then, we present our space lower bounds in Section~\ref{sec:lbs}. 
We present our $\Omega(n)$ space lower bound in Subsection~\ref{sec:lb-exact}. Then, in Subsection~\ref{sec:lb-approx}, 
we prove a space lower bound for approximation algorithms. %Finally, in Section~\ref{sec:open-problems}, we conclude with 
%open problems.

\textit{Missing Proofs.} Due to space restrictions, many proofs have been moved to the appendix. Lemmas and theorems with
deferred proofs are marked with $(*)$.

%\section{Approximation Algrithms}

\section{The \textsc{Probe} Algorithm} \label{sec:probe}

\begin{minipage}{9.7cm}
 An important building block for our algorithms is \textsc{Probe} (Algorithm~\ref{algo:probe}), 
and its variant \textsc{ProbeB} (not explicitly shown). These algorithms have been used in previous works on this problem, e.g. \cite{i91}.
\textsc{Probe} takes parameters $B$ and $p$, makes one pass over the input
stream and sets up partition separators such that partitions do not exceed a weight of $B$ but are of maximal
size. \textsc{ProbeB} performs the same task as \textsc{Probe}, but
it does not store the actual separators, and returns only a boolean value indicating whether the algorithm succeeded or failed.

We state now upper and lower bounds on the optimal bottleneck value $B^*$. Then, we use these bounds in order
to derive a bound on the space complexity of \textsc{Probe} and
\end{minipage}
\hfill
\begin{minipage}[c]{5.7cm}
\vspace{-1.5cm}
 \begin{algorithm}[H]
\caption{\textsc{Probe}($B, p$) \label{algo:probe}}
\begin{algorithmic}
 \STATE $I \gets 1$ \COMMENT{current element index}
 \STATE $P \gets 1$ \COMMENT{current separator index}
 \STATE $W \gets 0$ \COMMENT{current partition weight}
 \WHILE{input stream not empty}
  \STATE \textbf{if} $P > p$ \textbf{then} FAIL \textbf{end if}
  \STATE $x \gets $ next integer from stream
  \STATE $I \gets I+1$
  \STATE \textbf{if} $x> B$ \textbf{then} FAIL \textbf{end if}
  \STATE \textbf{if} $W+x \leq B$ \textbf{then} $W \gets W+x$
  \STATE \textbf{else} $s_P \gets I$, \ $P \gets P+1$, \ $W \gets x$  
  \STATE \textbf{end if}% \COMMENT {new interval started}
\ENDWHILE
\RETURN $(1, s_1, \dots, s_{p-1}, I)$
\end{algorithmic}
\end{algorithm}
\end{minipage}

\vspace{0.1cm}

\noindent \textsc{ProbeB}. Finally, we show how \textsc{Probe} and \textsc{ProbeB}
can be used to obtain a $(1+\epsilon)$-approximation.
%\noindent We summarize the properties of \textsc{Probe} and \textsc{ProbeB} in the following easily verifiable lemma.
In the following, let $S = \sum_i X_i$ denote the weight of the entire input integer sequence. %We observe the following upper and lower bounds on the optimal bottleneck value:

\begin{lemma} \label{lem:opt-bottleneck-bound}
Let $B^*$ denote the bottleneck value of an optimal partitioning. Then:
 \begin{eqnarray*}
\max \left\{ \lc \frac{S}{p} \rc, m \right\} \le & B^* & \le \lf \frac{S + (p-1)  m}{p} \rf <  \lf \frac{nm}{p} + m \rf.
 \end{eqnarray*}
\end{lemma}
\begin{proof}
For the lower bound $\lc \frac{S}{p} \rc$, observe that the weight of each partition is at most $B^*$, 
so their sum is at most $p \cdot B^*$. The integer $m$ is a trivial lower bound since an element of weight $m$ has to be 
part of some partition.

For the upper bound $\lf \frac{S + (p-1)  m}{p} \rf$, we construct a partitioning that fulfills this property. 
Assume that we know $S$ and $m$ in advance. 
Partition the stream greedily, placing a separator when the weight of the current partition is at least $B = \frac{S -m}{p}$. 
The weight of the current partition is thus at most $\lf B \rf + m =  \lf \frac{S + (p-1) \cdot m}{p} \rf$. After placing 
the separators $s_0, \dots, s_{p-1}$, the sum of the remaining elements (the weight of the last partition) is at most 
$S - (p-1)B = \frac{S + (p-1) \cdot m}{p}$. 

For the upper bound $\lf \frac{nm}{p} + m \rf$, note that $S \le nm$, and therefore, 
$\lf \frac{S + (p-1)  m}{p} \rf < \lf \frac{nm}{p} + m \rf$.
This implies the result. \qed
\end{proof}

The following lemma on the space requirements 
is easily verifiable and uses the previous bounds on the optimal bottleneck value $B^*$
of Lemma~\ref{lem:opt-bottleneck-bound}.

\begin{lemma} \label{lem:probe}
 \textsc{Probe}($B, p$) and \textsc{ProbeB}($B, p$) succeed if and only if the optimal bottleneck value is 
smaller or equal to $B$. \textsc{Probe} uses space $\Order(p \log n + \log B + \log m) = \Order(\log(mn^p))$ and \textsc{ProbeB} uses
space $\Order( \log p + \log B + \log m) = \Order(\log(mn))$. \qed
\end{lemma}

We show now that if $S$ is known in advance, using Lemma~\ref{lem:opt-bottleneck-bound}, \textsc{Probe} (respectively \textsc{ProbeB}) 
can be used to obtain a $(1+\epsilon)$-approximation algorithm for \textsc{Part} (resp. \textsc{PartB}).
As already mentioned in the introduction, this result is obtained by running $\Theta(\frac{\log p}{\log (1 + \epsilon)} )$
copies of \textsc{Probe} in parallel. For details, see the proof of Theorem~\ref{thm:probe} in the appendix.

\begin{theorem} \label{thm:probe}
 For any positive $\epsilon = \Order (1)$, if $S$ is known in advance, then by running $\Theta( \log(p) / \epsilon )$ copies of
\textsc{Probe} (resp.\ \textsc{ProbeB}) we can obtain a $(1+\epsilon)$-approximation algorithm for \textsc{Part} (resp.\ \textsc{PartB}).
The space requirements are 
\begin{itemize}
 \renewcommand{\labelitemi}{$\bullet$}
 \item $\Order \left( \log(p) \log (mn^p) / \epsilon \right)$ for \textsc{Part}, and 
 \item $\Order \left(\log(p) \log (mn) / \epsilon  \right)$ for \textsc{PartB}.
\end{itemize}
\end{theorem}
\begin{proof}
Let C = $\lc \frac{\log{p}}{\log(1+\epsilon)} \rc$. We run $C + 1$ copies of 
\textsc{Probe} (resp. \textsc{ProbeB}) in parallel, with bottleneck values 
$\frac{S}{p} (1 + \epsilon)^i$ for 
$i \in \{ 0, 1, \dots, C \}$. %where $C'$ is a constant that only depends on $C$. 
We return the successful partitioning with the smallest bottleneck value. Note that 
\begin{eqnarray*}
\frac{S}{p} (1 + \epsilon)^C \ge \frac{S}{p} \cdot p = S. 
\end{eqnarray*}
Let $B^*$ denote the optimal bottleneck value. Since $B^* \leq S$, there is always at least one run of 
\textsc{Probe} (resp. \textsc{ProbeB}) that succeeds, due to Lemma~\ref{lem:probe}. 
Let $B=\frac{S}{p} (1 + \epsilon)^{i'}$ be the returned bottleneck value.
Suppose first that $i' = 0$. Then $B = S/p$ and since $B^* \ge S/p$, we found the optimum. Otherwise $i' > 0$. Then
\begin{eqnarray*}
 \frac{S}{p} (1 + \epsilon)^{i'-1} \le B^* \le \frac{S}{p} (1 + \epsilon)^{i'} = B,
\end{eqnarray*}
and therefore $B^* (1+\epsilon) \ge B$ which proves the approximation ratio.

Observe that $C = \Theta (\log(p) / \epsilon)$ for any positive $\epsilon= \Order(1)$. The largest bottleneck value for which we run \textsc{Probe} (resp. \textsc{ProbeB})
is $\Order({mn})$. For an upper bound on the total space requirement, we multiply the maximal space requirement of a single copy of \textsc{Probe} or \textsc{ProbeB} (Lemma~\ref{lem:probe}) by the number of copies $C+1$. The result follows. \qed
\end{proof}

Finally, if $S$ is unknown to the algorithm but $m$ and $n$ are known, then the following holds:
\begin{theorem} \label{thm:probe2}
 For any positive $\epsilon = \Order (1)$, if $m$ and $n$ are known in advance, then by running 
$\Theta( \log(mn) / \epsilon )$ copies of \textsc{Probe} (resp.\ \textsc{ProbeB}) we can obtain a 
$(1+\epsilon)$-approximation algorithm for \textsc{Part} (resp.\ \textsc{PartB}).
The space requirements are
$\Order \left( (p \log^2 n + \log^2 m) / \epsilon \right)$ for \textsc{Part}, and 
$\Order \left(\log^2(mn) / \epsilon  \right)$ for \textsc{PartB}.
%\begin{itemize}
% \renewcommand{\labelitemi}{$\bullet$}
% \item $\Order \left( (p \log^2 n + \log^2 m) / \epsilon \right)$ for \textsc{Part}, and 
% \item $\Order \left(\log^2(mn) / \epsilon  \right)$ for \textsc{PartB}.
%\end{itemize}
\end{theorem}
The proof of Theorem~\ref{thm:probe2} is omitted since it is essentially equivalent to the proof of Theorem~\ref{thm:probe}
using the initial search space $\{m, m + 1, \dots, mn \}$.

%%%%

%\textit{Remark.} The idea of using \textsc{Probe} in order to obtain a 
%$(1+\epsilon)$-approximation is not new and
%was for instance used in \cite{i91} where \textsc{Probe} is executed multiple times sequentially in order to 
%perform a binary search on the optimal bottleneck value. Note that in Theorem~\ref{thm:probe}, we 
%run all copies of \textsc{Probe} simultaneously and independently from each other. %Nevertheless, we consider 
%this algorithm only as a baseline against which we compare our algorithms.

\section{The \textsc{ProbeExt} Algorithm} \label{sec:probe-ext}
\begin{minipage}{9.3cm}
In this section, we present a one-pass streaming algorithm that only requires the knowledge
of $m$ in advance. We denote this algorithm by \textsc{ProbeExt}, and similar to the \textsc{Probe}
algorithm, we introduce a counterpart \textsc{ProbeExtB} that does not store partition 
boundaries. \textsc{ProbeExt} receives $m$ and a real number $0 \le \alpha < 1$ as parameters,
and initially tries to set up maximal partitions of size at most $B = m (1+ \alpha)$. We discuss the 
actual purpose of $\alpha$ later, however, we mention that the choice of $\alpha$ does not 
affect the approximation factor of the algorithm. 
%For an easier understanding of the algorithm
%you can think of $\alpha$ to be $0$. 
Let $B^*$ denote the optimal bottleneck value. If $B^* > m (1+ \alpha)$ then
\textsc{ProbeExt} will reach a state where all $p$ partitions are set up,
while there are still integers in the input stream. 

In this situation, we \textit{merge}
all adjacent partitions $i$ and $i+1$ for odd $i$. In so doing, we create $\lf p/2 \rf + 1$
new partitions, each with weight at most $2m(1+\alpha)$. We double the current bottleneck value $B$ from $m(1+\alpha)$ to $2m(1+\alpha)$ 
\end{minipage}
\hfill
\begin{minipage}[c]{6.3cm}
\vspace{-1.5cm}
\begin{algorithm}[H]
\caption{\textsc{ProbeExt}($m, p, \alpha$)}
\begin{algorithmic}
 \STATE $I \gets 1$ \COMMENT{current element index}
 \STATE $P \gets 1$ \COMMENT{current separator index}
 \STATE $W \gets 0$ \COMMENT{current parition weight}
 \STATE $B \gets m (1 + \alpha)$ \COMMENT{curr.\ bottleneck value}
 \WHILE{input stream not empty}
  \STATE $x \gets $ next integer from stream
  \STATE $I \gets I+1$
  \STATE \textbf{if} $W+x \leq B$ \textbf{then} $W \gets W+x$ % \textbf{end if}
  \STATE \textbf{else if} $P < p$ \textbf{then} 
  \STATE $\quad$ $s_P \gets I$, \ $P \gets P+1$, \ $W \gets x$
  \STATE \textbf{else} \COMMENT{merge adjacent partitions}
  \STATE $\quad$ $s_P \gets I$, \ $B \gets 2B$, \ $P \gets \lf \frac{p}{2} \rf + 1$
  \STATE $\quad$ \textbf{for} $i = 1 \dots \frac{p}{2}$ \textbf{do} $s_i \gets s_{2i}$ \textbf{end for}
%  \STATE $s_i \gets s_{2i}$, for all $1 \leq i \leq \frac{p}{2}$
  \STATE $\quad$ \textbf{if} $p$ even \textbf{then} $W \gets x$ 
  \STATE $\quad$ \textbf{else} $W \gets W+x$ \textbf{end if}
  \STATE \textbf{end if}
\ENDWHILE
\RETURN $B$, $(1, s_1, \dots, s_{p-1}, I)$
\end{algorithmic}
\end{algorithm}
\end{minipage}

\vspace{0.1cm}
\noindent 
and we continue setting up partitions. We perform these steps repeatedly until we reach the end of the stream, and 
we obtain a bottleneck value of $2^i (1+\alpha) m$,  where $i$ denotes the number of merge operations that
occurred during the execution of the algorithm.
We summarize the space requirements of \textsc{ProbeExt} and \textsc{ProbeExtB} in the following lemma.

\begin{lemma} \label{lem:space-probe-ext}
 A run of \textsc{ProbeExt} requires space $\Order(\log (mn^p))$, and a run of \textsc{ProbeExtB} 
requires space $\Order(\log (mn))$.
\end{lemma}
\begin{proof}
 \textsc{ProbeExt} stores the separators, which accounts for $\Order(p \log n)$ space. Furthermore,
it stores the variable $B$ which is bounded by the bottleneck value of the partitioning it creates. As we show later that the algorithm is a constant factor approximation, this value is in the order of the optimal
bottleneck value, which in turn is bounded by $\Order( \frac{mn}{p})$, see Lemma~\ref{lem:opt-bottleneck-bound}.
Therefore, we obtain the bound $\Order(p \log n + \log (\frac{mn}{p}) ) = \Order(\log (mn^p))$.

\textsc{ProbeExtB} does not store the separators. Therefore, its space requirement is bounded by the 
optimal bottleneck value $\Order(\log (\frac{mn}{p})) = \Order( \log (mn) )$.  \qed
\end{proof}

In the remainder of this section, we show that if the optimal bottleneck value $B^*$ is large compared to $m$, then the algorithm 
is close to a $2$-approximation (see Lemma~\ref{thm:probeext}). 
We use this fact
in Section~\ref{sec:known-m} to obtain a $(1+\epsilon)$-approximation algorithm.

\begin{lemma} \label{lem:stream-weight}
 Suppose that \textsc{ProbeExt} (or \textsc{ProbeExtB}) performs $i$ merge operations. Then the weight of the input 
stream is at least 
\begin{eqnarray*}
\frac{pm}{2} \left( 2^i (1 +\alpha) - \alpha - i \right) - \frac{m}{2} (i+\alpha).
\end{eqnarray*}
\end{lemma}
\begin{proof}
 We develop a lower bound on the total weight of the stream after $i$ merge operations have been
 executed, and we denote this lower bound by $LB(i)$.
 
 Consider the situation just before the first merge operation. Denote by $w_j$ the weight of the $j$th partition. Note that
 for all $j$ we have $w_j + w_{j+1} > m(1+\alpha) $, otherwise the algorithm would have created a single partition instead of the two adjacent partitions $j$ and $j+1$. Thus, if $p$ is even, we have $\sum_j w_j \ge \frac{1}{2} p m(1+\alpha)$ and if $p$ is odd, we have 
$\sum_j w_j \ge \frac{1}{2} (p-1) m(1+\alpha)$. To unify the analysis for the even and the odd case, we set 
\begin{eqnarray*}
LB(1) = \frac{1}{2} (p-1) m(1+\alpha). 
% LB(1) = \begin{cases} \frac{1}{2} p m(1+\alpha) & \mbox{if $p$ is even} \\
%  \frac{1}{2} (p-1) m(1+\alpha) & \mbox{if $p$ is odd}. \end{cases}
\end{eqnarray*}
%Note that if $p$ is odd, we only consider the weight of the first $p-1$ partitions since those partitions will 
%constitute the first $(p-1)/2$ partitions after the merge operation. 
Consider now the situation just before the $i$th merge operation, again denoting by $w_j$ the weight of the $j$th partition. 
If $p$ is even, then the weight of the first $p/2$ partitions is at least $LB(i-1)$. 
Clearly, each of the remaining $p/2$ partitions have weight of at least $\lc 2^{i-1} m(1+\alpha) - m \rc$, hence:
\begin{eqnarray*}
 \sum_j w_j = \sum_{j=1}^{p/2} w_j + \sum_{j=p/2+1}^{p} w_j \ge LB(i-1) + \frac{p}{2} ( 2^{i-1} m(1+\alpha) - m ).
\end{eqnarray*}
Suppose now that $p$ is odd. Then the weight of the first $(p-1)/2$ partitions is at least $LB(i-1) - \lf 2^{i-2} m (1 + \alpha) \rf$. 
The remaining $(p+1)/2$ partitions have a weight of at least $\lc 2^{i-1} m(1+\alpha) - m \rc$, and we obtain
\begin{eqnarray*}
 \sum_j w_j = \sum_{j=1}^{(p-1)/2} w_j + \sum_{j=(p+1)/2}^{p} w_j 
 & \ge & LB(i-1) - \lf 2^{i-2} m (1 + \alpha) \rf + \frac{p+1}{2} \lc 2^{i-1} m(1+\alpha) - m \rc \\
 & \ge & LB(i-1) + \frac{p}{2} ( 2^{i-1} m (1+\alpha) - m) - \frac{1}{2} m.
\end{eqnarray*}
In order to treat the even and the odd case at the same time, we set
\begin{eqnarray*}
 LB(i) = LB(i-1) + \frac{p}{2} ( 2^{i-1} m (1+\alpha) - m) - \frac{1}{2} m,
\end{eqnarray*}
and we eliminate the recursion:
\begin{eqnarray*}
%%% & = & \sum_{j=2}^i  {\left( \frac{p}{2} ( 2^{j-1} m (1+\alpha) - m) - \frac{1}{2} m  \right) } + \frac{1}{2} (p-1) m(1+\alpha) \\
%%% & = & \frac{pm}{2} \sum_{j=2}^i {\left( 2^{j-1} (1+\alpha) - 1 \right) } - \frac{m}{2}(i-1) + \frac{pm}{2} (1+\alpha)  - \frac{m}{2}(1+\alpha) \\
% & = & \frac{pm}{2} \left( (1+\alpha) \left( \sum_{j=1}^{i-1} 2^j \right) - (i-1) \right) + \frac{pm}{2} (1+\alpha) - \frac{m}{2}(i + \alpha ) \\
% & = & \frac{pm}{2} \left( (1+\alpha) (2^i - 2) -i + 2 + 1 + \alpha \right) - \frac{m}{2} (i-1+\alpha) \\
%%% & = & \frac{pm}{2} \left( 2^i (1 +\alpha) -2 - 2 \alpha -i +2 + \alpha \right) - \frac{m}{2} (i+\alpha) \\
 LB(i) & = &  \sum_{j=2}^i {\bigl(LB(j) - LB(j-1) \bigr)} + LB(1) 
 = \frac{pm}{2} \left( 2^i (1 +\alpha) - \alpha - i \right) - \frac{m}{2} (i+\alpha).
\end{eqnarray*} \qed
%which proves the result.
 
\end{proof}

\begin{lemma}\label{lem:merge-approx}
 Suppose that \textsc{ProbeExt} (or \textsc{ProbeExtB}) performs $i$ merge operations, for $i \ge 2$. Then for any $0 \le \alpha < 1$, 
\textsc{ProbeExt} (resp. \textsc{ProbeExtB}) has an approximation factor of at most
\begin{eqnarray*}
 2 + \frac{2 (\alpha + i)}{ 2^{i-1} (1+\alpha) - i - \alpha } .
\end{eqnarray*}
\end{lemma}
\begin{proof}
 Let us denote the bottleneck value of the solution returned by \textsc{ProbeExt} (resp. \textsc{ProbeExtB}) by $B = 2^{i} m (1+\alpha)$, and let $B^*$ denote the optimal bottleneck value. By Lemma~\ref{lem:stream-weight}, the total weight $S$ of the stream is at least 
$\frac{pm}{2} \left( 2^i (1 +\alpha) - \alpha - i \right) - \frac{m}{2} (i+\alpha) $, and $B^*$ is at least a $p$-fraction of S (Lemma~\ref{lem:opt-bottleneck-bound}).
The approximation factor of \textsc{ProbeExt}
(resp. \textsc{ProbeExtB}) can be 
bounded as follows:
\begin{eqnarray*}
\frac{B}{B^*} & \le & \frac{2^i m (1+\alpha)}{\frac{1}{p} \cdot \left( \frac{pm}{2} \left( 2^i (1 +\alpha) - \alpha - i \right) - \frac{m}{2} (i+\alpha) \right)} \le 2 + \frac{2 (\alpha + i)}{ 2^{i-1} (1+\alpha) - i - \alpha }\ .  
%  & = &  \frac{2^{i+1}(1+\alpha)}{ 2^i (1+\alpha) - i - \alpha  - {(i+\alpha)}/{p}} \\
% & = & 2 + \frac{2 (\alpha + i + {(i+\alpha)}/{p})}{ 2^i (1+\alpha) - i - \alpha  - {(i+\alpha)}/{p}}\\
\end{eqnarray*}
\qed
\end{proof}
We conclude with the following result:

\begin{lemma} \label{thm:probeext}
 For any $0 \le \alpha < 1$ \textsc{ProbeExt} (or \textsc{ProbeExtB}) 
is a $(2+\epsilon)$-approximation algorithm if the optimal bottleneck 
value satisfies $B^* > m / \epsilon^2$, assuming $0 < \epsilon \le 1/64$.
\end{lemma}
\begin{proof}
 By Lemma~\ref{lem:merge-approx}, \textsc{ProbeExt} is a $\left( 2 + \frac{2 (\alpha + i)}{ 2^{i-1} (1+\alpha) - i - \alpha } \right)$- approximation algorithm if $i$ merge operations have been executed. We have:
\begin{eqnarray*}
 \frac{2(\alpha + i)}{ 2^{i-1} (1+\alpha) - i - \alpha }  \le \frac{2(i + 1)}{ 2^{i-1} - (i + 1)} \le \frac{1}{2^{i/2}} = \epsilon.
\end{eqnarray*}
The first inequality uses the bounds on $\alpha$. To make sure that all quantities are positive and the second inequality also holds, we require $i \ge 12$. The last equality gives $i = 2 \log (\frac{1}{\epsilon})$, and our previous bound on $i$ forces $\epsilon \le 1/64$. 
Under these conditions, \textsc{ProbeExt} is a $(2+\epsilon)$-approximation algorithm. 
Since $B \ge B^*$, \textsc{ProbeExt} must perform $i$ merge operations if 
%\begin{eqnarray*}
$B^* > m(1+\alpha) 2^{i-1}$.
%\end{eqnarray*}
Since 
%\begin{eqnarray*}
$m(1+\alpha) 2^{i-1} < m 2^{i} = \frac{m}{\epsilon^{2}}$,
%\end{eqnarray*}
the condition $B^* > m / \epsilon^{2}$ is a sufficient one.  \qed
\end{proof}

%We observe that the condition $B^* > m / \epsilon^2$ of Lemma~\ref{thm:probeext} can be relaxed 
%slightly, and $B^* > m / \epsilon^\gamma$ is sufficient, for any $\gamma>1$. 
%In the following, we apply this lemma obtaining a space bound that is logarithmic in the $\frac{m}{\epsilon^\gamma}$.
%Choosing a different $\gamma < 2$, however, would only yield an improvement by a constant factor in the overall space bound.
%For this reason and for the sake of simplicity, we stick to $\gamma = 2$.

\section{$(1+\epsilon)$-approximation for Known $m$} \label{sec:known-m}

In this section, we present a $(1+\epsilon)$-approximation algorithm for \textsc{Part} and \textsc{PartB}, 
using as building blocks the \textsc{Probe} and \textsc{ProbeExt} algorithms presented in 
Sections~\ref{sec:probe} and \ref{sec:probe-ext}. We assume that the maximum $m$ of the sequence is known in advance.

%Observe that $m$ can also be an upper bound on the maximum, and need not even be a close approximation. The space usage depends on the value of $m$, however, for the algorithm to work correctly, we only need that $m \le B^*$. Therefore $m$ can be as large as $\frac{S}{p}$, and the algorithm will still find a $(1+\epsilon)$-approximation. \lk{this claim seems true but needs to be checked carefully - also are there any implications?}

%%%We fix an arbitrarily small positive constant $\gamma$, with the restriction that $\epsilon \le \epsilon'$, where $\epsilon'$ depends only on $\gamma$. 
%The constant $\gamma$ is only required for the analysis of the space usage of the algorithm. 

\begin{algorithm}
\caption{$(1+\epsilon)$-Approximation For Known $m$ \label{algo:known-m}}
 \begin{algorithmic}
   \STATE $\delta \gets \frac{\epsilon}{1 + \frac{1}{2} \epsilon}$
   \STATE \textbf{do in parallel } \COMMENT{in one pass}
   \STATE $\quad$ \textsc{Probe}$(2^i (1 + \epsilon)^jm, p)$ for all $i \in \left\{0, 1, \dots, \lc \log ( 1/{\delta}^{2}) \rc \right\}$, $j \in \left\{0, 1, \dots, \lc \frac{1}{\log (1 + \epsilon)} \rc \right\}$
   \STATE $\quad$ \textsc{ProbeExt}($m, p, (1+ \frac{1}{2} \epsilon)^j - 1)$, for all $j \in \left\{0, 1, \dots, \lc \frac{1}{\log(1+\frac{1}{2} \epsilon) } \rc \right\}$
   \STATE \textbf{end do}
   \RETURN partitioning with smallest bottleneck value
 \end{algorithmic}
\end{algorithm}

Algorithm~\ref{algo:known-m} runs multiple copies of the the \textsc{Probe} algorithm and multiple copies of the 
\textsc{ProbeExt} algorithm in parallel. We argue that if the optimal bottleneck value $B^*$ is sufficiently large 
then one run of the \textsc{ProbeExt} algorithm will return a $(1+\epsilon)$-approximation. If $B^*$ is small,
then a run of the \textsc{Probe} algorithm will return a $(1+\epsilon)$-approximation. 

\begin{theorem} \label{thm:known-m}
 For any $\epsilon < 1/64$, Algorithm~\ref{algo:known-m} is a $(1+\epsilon)$-approximation streaming algorithm 
for \textsc{Part} using space $\Order \left( \frac{1}{\epsilon} \log(\frac{1}{\epsilon}) \log ({mn^p}) \right)$. 
The analogous algorithm for \textsc{PartB}  is a $(1+\epsilon)$-approximation streaming algorithm using space 
%$\Order \left( \frac{1}{\epsilon} \log \frac{1}{\epsilon} \cdot \log \frac{mp}{\epsilon} 
%+ \frac{1}{\epsilon}  \log (mn) \right) = $
$\Order \left( \frac{1}{\epsilon} \log(\frac{1}{\epsilon}) \log ({mn}) \right)$.
%+ \frac{1}{\epsilon}  \log n \right)$.\lk{*}
\end{theorem}
\begin{proof}
We distinguish two cases depending on the magnitude of the optimal bottleneck value $B^*$. 
%Let $B^*$ be the optimal bottleneck value. We carry out a case distinction on the value of $B^*$ 
%in order to prove the approximation guarantee of Algorithm~\ref{algo:known-m}.
In the following, $\delta = \frac{\epsilon}{1+\frac{1}{2} \epsilon}$ as in Algorithm~\ref{algo:known-m}.
\begin{enumerate}
 \item $B^* \le {m}/{\delta}^{2}: $ We show that one of the runs of \textsc{Probe} is successful and 
returns a partitioning with bottleneck value $B$ such that $B \le B^* (1 + \epsilon)$. We run \textsc{Probe}
with bottleneck values $2^i (1 + \epsilon)^j m$, and since there is a run with $i = \lc \log ( 1/{\delta}^{2}) \rc$, 
there is at least one successful run of \textsc{Probe} with a bottleneck value of at most ${m}/{\delta}^{2}$. 
Let $B = 2^{i'} (1 + \epsilon)^{j'} m$ denote the smallest bottleneck value of a successful run for values $i', j'$.
Suppose that $j' > 0$. Then the run with bottleneck value $2^{i'} (1 + \epsilon)^{j'-1} m$ failed, and
therefore 
\begin{eqnarray*}
 B = 2^{i'} (1 + \epsilon)^{j'} m \ge B^* > 2^{i'} (1 + \epsilon)^{j'-1} m,
\end{eqnarray*}
which implies $B \le (1 + \epsilon) B^*$. Suppose now that $j' = 0$ and $i' > 0$. Then $B = 2^{i'} m$, and the run with bottleneck value $2^{i'-1} (1+\epsilon)^{\lc \frac{1}{\log (1 + \epsilon)} \rc -1} m$ failed, 
and therefore
\begin{eqnarray*}
 B = 2^{i'} m \ge B^* > 2^{i'-1} (1+\epsilon)^{\lc \frac{1}{\log (1 + \epsilon)} \rc -1} m,
\end{eqnarray*}
which also implies $B \le (1 + \epsilon) B^*$. If $i'=j'=0$, then the algorithm found an optimal solution with bottleneck value $m$.

Since for $\epsilon = \Order(1)$ we have $\log (1 + \epsilon) = \Order(\epsilon)$ and $\delta = \Theta(\epsilon)$, 
the space requirement for the runs of \textsc{Probe} is 
$\Order \left( \frac{1}{\epsilon} \log (\frac{1}{\epsilon}) \log ({mn^p})  \right)$, 
and if we run \textsc{ProbeB} the space requirement is 
$\Order \left( \frac{1}{\epsilon} \log (\frac{1}{\epsilon}) \log (\frac{mp}{\epsilon})   \right)$. 

 \item $B^* > {m}/{\delta}^{2}: $ 
 %Since the bottleneck value is at least $\frac{m}{\delta}$, 
 By 
Lemma~\ref{thm:probeext}, \textsc{ProbeExt} and \textsc{ProbeExtB} are
$(2 + \delta)$-approximation algorithms for any $\alpha$. Let $B = 2^{i'} (1+ \frac{1}{2} \epsilon)^{j'} m$ be the smallest value output by any
of the \textsc{ProbeExt} runs, for some values of $i'$ and $j'$.
Suppose that $j' > 0$. Then the run with $j = j' - 1$ reports the bottleneck value 
$2^{i' + 1} (1+ \frac{1}{2} \epsilon)^{j'-1} m$. Clearly, it cannot return 
$2^{k} (1+ \frac{1}{2} \epsilon)^{j'-1} m$ for $k \le i'$ since $B$ is the smallest
returned value. On the other hand, it cannot return a bottleneck with $k \ge i' + 2$ since then it would have an approximation ratio larger than $2+\delta$, contradicting Lemma~\ref{thm:probeext}.
\begin{eqnarray*}
 2^{i' + 2} (1+ \frac{1}{2} \epsilon)^{j'-1} m = \frac{4} {1+\epsilon} \cdot B \ge \frac{4} {1+\epsilon} B^* > (2 + \delta) B^*.
\end{eqnarray*}
Thus, the run with $j = j' - 1$ returns the bottleneck value $2^{i' + 1} (1+ \frac{1}{2} \epsilon)^{j'-1} m$.
Since this is a $2+\delta$ approximation, we obtain
\begin{eqnarray*}
 (2 + \delta) B^* & \ge & 2^{i' + 1} (1+ \frac{1}{2} \epsilon)^{j'-1} m = \frac{2}{1+\frac{1}{2} \epsilon} B \Rightarrow \\
 B  & \le & \frac{(2+ \delta) (1 + \frac{1}{2} \epsilon) }{2} B^* = (1+\epsilon) B^*.
\end{eqnarray*}
 
Suppose now that $j'=0$ and $i' > 0$. Consider the run for $j = \lc \frac{1}{\log(1+\frac{1}{2} \epsilon) } \rc - 1$. 
By a similar argument as before, the run outputs the bottleneck value $2^{i'}(1+ \frac{1}{2} \epsilon)^{\lc \frac{1}{\log(1+\frac{1}{2} \epsilon) } \rc - 1}m = B (1+ \frac{1}{2} \epsilon)^{\lc \frac{1}{\log(1+\frac{1}{2} \epsilon) } \rc - 1}$. This implies that
\begin{eqnarray*}
 B^* \ge \frac{B (1+ \frac{1}{2} \epsilon)^{\lc \frac{1}{\log(1+\frac{1}{2} \epsilon) } \rc - 1}}{2+\delta} 
 \ge \frac{B (1+ \frac{1}{2} \epsilon)^{ \frac{1}{\log(1+\frac{1}{2} \epsilon) }  - 1}}{2+\delta},
\end{eqnarray*}
which also implies that $B \le (1+\epsilon) B^*$. Finally, if $i' = j' = 0$, then the algorithm did not
perform a merge operation and found an optimal solution with bottleneck value $m$. 

Since $ \frac{1}{\log(1 + \frac{1}{2} \epsilon)} = \Order(\frac{1}{\epsilon})$, the space requirement for 
the runs of \textsc{ProbeExt} is $\Order \left(\frac{1}{\epsilon}  \log (mn^p) \right)$,
and if we run \textsc{ProbeExtB} the space requirement is $\Order \left(\frac{1}{\epsilon}  \log (mn)  \right)$.
\end{enumerate}

For \textsc{Part}, the space requirements are dominated by the runs of the \textsc{Probe} algorithm. For \textsc{PartB},
we obtain space $\Order \left( \frac{1}{\epsilon} \log(\frac{1}{\epsilon}) \log(\frac{mp}{\epsilon}) 
+ \frac{1}{\epsilon}  \log (mn) \right)$, and using $p \le n$ this simplifies to
$\Order \left( \frac{1}{\epsilon} \log(\frac{1}{\epsilon}) \log(\frac{mn}{\epsilon}) \right)$.
\qed
\end{proof}

\section{Algorithms for Unknown $m$ } \label{sec:unknown-m}

%\subsection{$2$-approximation for \textsc{PartB} for Unknown $m$ } \label{sec:unknown-m-partb}
\begin{minipage}{8cm}
In this section, we present simple $2$-approximation algorithms for \textsc{Part} (resp.\ \textsc{PartB}) that 
do not require the knowledge of any parameter in advance.
%the knowledge of $m$ or $S$ in advance (Algorithm~\ref{algo:unknown-m-partb}). 
%We obtain the following theorem:

%\begin{theorem} \label{thm:unknown-m-partb}
% Algorithm~\ref{algo:unknown-m-partb} is a $2$-approximation algorithm for \textsc{PartB}
%and uses $\Order(\log(mn))$ space.
%\end{theorem}

%\end{minipage}
%\hfill
%\begin{minipage}[c]{6.5cm}
%\vspace{-0.8cm}
%\begin{algorithm}[H]
%\caption{$2$-Approx. for \textsc{PartB} \label{algo:unknown-m-partb}}
% \begin{algorithmic}
%   \STATE \textbf{In one pass:}
%   \STATE $\quad$ Compute total weight $S = \sum_i X_i$
%   \STATE $\quad$ Compute maximum $m = \max_i X_i$
%   \RETURN $\max \{m, \frac{S}{p} \} + m$
% \end{algorithmic}
%\end{algorithm} 
%\end{minipage}

%Suppose that $m \ge \frac{S}{p}$. Then the algorithm returns the value $2m$ which proves the approximation factor.
%Suppose now that $\frac{S}{p} > m$. Then by Lemma~\ref{lem:opt-bottleneck-bound}, $B^* \le \frac{S}{p} + m$. Since
%$\frac{S}{p} > m$, the algorithm outputs a value that is at most $2 \frac{S}{p}$ which proves the result.

%Concerning the space requirements, the weight of the stream is at most $m n$, and therefore space $\Order(\log (mn))$ is enough. \qed
%\end{proof} 

%\subsection{$2$-approximation for \textsc{Part} for Unknown $m$} \label{sec:unknown-m-part}
%We present a $2$-approximation algorithm for \textsc{Part} that does not require 
%the knowledge of any parameter in advance (Algorithm~\ref{algo:unknown-m-part}). 

Our algorithm for \textsc{Part} works as follows: Suppose that the algorithm has seen the elements $X_1, \dots, X_i$ and it has partitioned them into $p$ parts with weights $w_1, \dots, w_p$. If the algorithm now reads the input $x = X_{i+1}$, it will run the \textsc{Probe} algorithm on the sequence $w_1, w_2, \dots, w_p, x$ with a bottleneck value $B$ that is at most twice the optimum for a partitioning of $X_1, \dots, X_{i+1}$ into $p$ parts. See Algorithm~\ref{algo:unknown-m-part} and Theorem~\ref{thm:unknown-m-part} for further details. The algorithm for \textsc{PartB} is even simpler, and is described in Theorem~\ref{thm:unknown-m-partb}.
\end{minipage} \hfill \begin{minipage}{7.5cm}
 %From Theorem~\ref{thm:unknown-m-part} it follows that \textsc{Probe} succeeds, guaranteeing an
%approximation ratio of $2$.
%Suppose that our algorithm has seen the prefix $X[1] \dots X[i]$ and 
%is about to treat the new item $x = X[i+1]$. First, our algorithm updates a value $B$ that is at most by a factor 
%$2$ larger than the optimal bottleneck value for a partitoning of $X[1] \dots X[i+1]$ into $p$ parts. Let 
%$W[1], \dots, W[p]$ denote the weights of the current partitioning (of $X[1] \dots X[i]$). Then, our
%algorithm runs the \textsc{Probe} algorithm on the integer sequence $W[1] W[2] \dots W[p] x$ with bottleneck value 
%$B$. In Theorem~\ref{thm:unknown-m-part}, we show that \textsc{Probe} always succeeds which guarantees 
%an approximation ratio of $2$. 

\vspace{-0.8cm}

\begin{algorithm}[H]
\caption{$2$-Approximation for \textsc{Part} \label{algo:unknown-m-part}}
 \begin{algorithmic}
   \STATE $w_i \gets 0$ for all $1 \le i \le p$ \COMMENT{partition weights}
   \STATE $I \gets 1$ \COMMENT{current element index}
   \STATE $S \gets 0$ \COMMENT{current total weight of input stream}
   \STATE $m \gets 0$ \COMMENT{current maximum}
   \WHILE{input stream not empty}
     \STATE $x \gets$ next integer from stream
     \STATE $I \gets I + 1$, $S \gets S + x$, $m \gets \max \{ m, x \}$
     \STATE $B \gets 2 \cdot \max \{ m, S / p \}$ \COMMENT{update bottleneck value}
     \STATE Run \textsc{Probe}$(B)$ on sequence $w_1, w_2, \dots, w_p, x$, 
     \STATE   \quad and store the new separators $s_0, \dots, s_p$ 
     \STATE   \quad and the new partition weights $w_1, \dots, w_p$ \label{line:192}
   \ENDWHILE 
   \RETURN $B, (1, s_1, \dots, s_{p-1}, I)$
 \end{algorithmic}
\end{algorithm}
\end{minipage}

\begin{theorem} \label{thm:unknown-m-part}
 Algorithm~\ref{algo:unknown-m-part} is a $2$-approximation algorithm for \textsc{Part} and uses space $\Order(p \log (mn) )$.
\end{theorem}
\begin{proof}
 First, suppose that the run of \textsc{Probe} in %Line~\ref{line:192} of 
 Algorithm~\ref{algo:unknown-m-part} succeeds in 
every iteration. Then, the last bottleneck value is %$B$ of the computed partitioning in the last run of \textsc{Probe} is 
$B = 2 \cdot \max \{ m, S / p \}$. By Lemma~\ref{lem:opt-bottleneck-bound}, we have 
$\max \{m, S/p \} \le B^* \le \max \{m, S/p \} + m$, and since $m \le \max \{m, S/p \}$ we have $\max \{m, S/p \} \le B^* \le 2 \cdot \max \{m, S/p \}$ which proves the approximation factor of $2$. 

Denote $w_{p+1}=x$. It remains to prove that the run of \textsc{Probe} always succeeds, i.e., that the 
optimal bottleneck value of the sequence $w_1, w_2, \dots w_p, w_{p+1}$ is at most $B = 2 \cdot \max \{ m, S / p \}$ in every round.
Indeed, if \textsc{Probe}$(B)$ does not succeed in creating $p$ partitions, then $w_i + w_{i+1}> B \ge 2 S / p$ must hold for all $1 \le i \le p$. But then:
%$W_i > 0$ must hold for all $1\le i \le p+1$. Furthermore, for all $1 \le i \le p$ we have  $W_i + W_{i+1}> B \ge 2 S / p$.
%\begin{eqnarray*}
% W_i + W_{i+1}> B \ge 2 S / p.
%\end{eqnarray*}
%We prove this statement by contradiction: Suppose that a run of \textsc{Probe}$(B)$ does not succeed, i.e., it would create
%more than $p$ partitions when run on the sequence $W_1, W_2, \dots W_p, x$.% with bottleneck value $B = 2 \cdot \max \{ m, S / p \}$.
%Let $W'_1, \dots, W'_{p'}$ denote the weights of the resulting partitions ($p' = p+1$). Then $W'_i > 0$ for all $1\le i \le p'$. Furthermore, for $1 \le i \le p' - 1$:
%It follows:
\begin{eqnarray*}
 S  = \sum_{i=1}^{p+1} {w_i} \ge \sum_{i=1}^{\lf (p+1)/2 \rf} \left( w_{2i - 1} + w_{2i} \right) > \lf (p+1)/2 \rf \cdot 2 S / p \ge S,
\end{eqnarray*}

a contradiction, which proves the correctness of the algorithm.
%Suppose that $p'$ is even. Then:
%\begin{eqnarray*}
% S = \sum_{i=1}^{p'/2} \left( W'_{2i - 1} + W'_{2i} \right) > \frac{p'}{2} \cdot 2 S / p > S.
%\end{eqnarray*}
%Similarly, if $p'$ is odd:
%\begin{eqnarray*}
% S = W'_{p'} + \sum_{i=1}^{(p'-1)/2} \left( W'_{2i - 1} + W'_{2i} \right) > \frac{p'-1}{2} \cdot 2 S / p = S,
%\end{eqnarray*}
%a blatant contradiction in both cases, which proves the correctness of the algorithm.
%which is a contradiction since the sum of partition weights cannot exceed the total weight of the input stream.
%a contradiction. 
%Concerning space requirements, the space is dominated by the space requirements for storing the weight of the $p$ partitions. We therefore have space $\Order(p \log(mn))$. 
The space requirement is dominated by the weights of the $p$ partitions, yielding the bound $\Order(p \log(mn))$. \qed
\end{proof}

\begin{theorem} \label{thm:unknown-m-partb}
 There exists a $2$-approximation algorithm for \textsc{PartB} that uses $\Order(\log(mn))$ space.
\end{theorem}

\begin{proof}
We simply compute in one pass the total weight $S$ and the maximum $m$, then output $\max \{m, \frac{S}{p} \} + m$.
By Lemma~\ref{lem:opt-bottleneck-bound} we have $\max \{m, \frac{S}{p} \} \le B^* \le \max \{m, \frac{S}{p} \}+m$. Hence, the approximation ratio is at most $1+m/\max \{m, \frac{S}{p} \} \le 2$.
The total weight of the stream is at most $m n$, therefore the space usage is $\Order(\log (mn))$. \qed
\end{proof}

\section{Space Lower Bounds} \label{sec:lbs}

\subsection{A Linear Space Lower Bound for Exact Algorithms} \label{sec:lb-exact}
In this section, we show that any possibly randomized exact streaming algorithm for either \textsc{Part} or \textsc{PartB}
that performs one pass over the input requires $\Omega(n)$ space. We show this by a reduction from the \textsc{Index} problem in 
one-way two-party communication complexity.

\begin{definition}[\textsc{Index} Problem]
  Let $S = (S_1, \dots, S_N)$ where $S \in \{0,1\}^N$, and let $I \in \{1,\dots,N\}$. 
  Alice is given $S$, Bob is given $I$.
Alice sends message $M$ to Bob and upon reception Bob outputs $S_I$.  
\end{definition}

We consider a version of \textsc{Index} where the index $I$ is chosen from the set $\{\lc N/2 \rc, \dots, N\}$ uniformly at random.
It is well-known \cite{kn97} that the one-way randomized communication complexity of \textsc{Index} is 
$\Omega(N)$,
and the modification in the input distribution restricting the index $I$ to be chosen from the set 
$\{\lc N/2 \rc, \dots, N\}$ does not change its hardness.

\begin{lemma}[Hardness of the Index Problem] \label{fact:hardness-AI}
 If $S$ is chosen uniformly at random from $\{0,1\}^N$, and $I$ is chosen uniformly at random from the set 
$\{\lceil N/2 \rceil, \dots, N\}$ and the failure probability of the protocol is at most $1/3$, then 
$\Exp_{S} |M| = \Omega(N)$.\ \ \qed 
\end{lemma}

\textit{Reduction.} 
Given a streaming algorithm $ALG$ that solves \textsc{Part} or \textsc{PartB} on a stream of length at most $3n$
using space $s$, we specify a protocol for an arbitrary instance $(S,I)$ of the one-way two-party communication 
problem \textsc{Index} with $|S| = n$, such that the message size is at most $s$.

Remember that Alice holds $S \in \{0,1\}^N$ and Bob holds
$I \ge \lceil N/2 \rceil$. Our protocol is the following: Alice generates the sequence $Y \in \{1,3\}^{2N}$ such that $Y_i = 2 \cdot S_{i/2} + 1$ for even $i$, and 
$Y_i = 4 - Y_{i+1}$ for odd $i$. Bob generates the sequence $Z = \underbrace{4 \dots 4}_{2I - N - 1}2$. 

Alice runs $ALG$ on the sequence 
$Y$ with the number of partitions $p=2$. Once $Y$ is entirely processed, she sends the resulting memory state of $ALG$ to Bob. Bob continues running $ALG$ on 
Alice's final memory state and feeds the sequence $Z$ into $ALG$. Observe that from the point of view of $ALG$ it is as if the input stream were the concatenation of $Y$ and $Z$. The message size of the protocol equals the space usage of $ALG$ after processing $Y$. If $ALG$ is an algorithm for \textsc{Part} then 
$ALG$ outputs the separator $s_1$ that separates
the two partitions. If $s_1$ is even then Bob outputs $0$, and if $s_1$ is odd then Bob outputs $1$.
If $ALG$ is an algorithm for \textsc{PartB} then $ALG$ outputs the optimal bottleneck value $B$.
If $B = 4I - 1$ then Bob outputs $0$, otherwise Bob outputs $1$.

We prove that the above protocol is correct (i.e.\ the value returned by Bob is $S_I$), which immediately yields the space lower bound.

\begin{theorem} \label{thm:lb-exact}
Any possibly randomized exact one-pass streaming algorithm for \textsc{Part} or \textsc{PartB} requires space $\Omega(n)$.
\end{theorem}
\begin{proof}
First observe that $\sum_i Y_i + \sum_i Z_i = 4 \cdot N + (2I - N - 1) \cdot 4 + 2 = 8I - 2$.
%Furthermore, since $\sum_i Y_i = 4 N$ and $N \ge I$, the optimal split of $Y \circ Z$ splits two elements
%of $Y$. 
Let $s_1^*$ denote the optimal split position. Suppose that a perfect balancing is achieved and the
optimal bottleneck value is $4I - 1$. Since for all $i=1,\dots,N$ we have $Y_{2i-1} + Y_{2i} = 4$,
this can only be achieved if $s_1^*$ is even and $Y_{s_1^*} = 1$ which implies that $S_{s_1^*/2} = S_{I} = 0$. Suppose now that a perfect balancing 
cannot be achieved. This can only happen if $s_1^*$ is odd and $Y_{s_1^*-1} = 3$ which implies that
$S_{(s_1^*-1) / 2} = S_I = 1$. 
Thus, the protocol is correct in both cases, and $ALG$ can be used to solve \textsc{Index}. Lemma~\ref{fact:hardness-AI} gives hence a lower bound for the 
space requirements of $ALG$. %Note that in the application of Lemma~\ref{fact:hardness-AI} we have $\log | \mathcal{U} | = 1$. 
\qed
\end{proof}

\subsection{$\Omega(\frac{1}{\epsilon} \log n)$ Space Lower Bound for Approximation Algorithms} \label{sec:lb-approx}
In this section, we prove an $\Omega( \frac{1}{\epsilon} \log n)$ space lower bound for
one-pass streaming algorithms for \textsc{Part} that compute a $(1+\epsilon)$-approximation.
We prove this lower bound in the one-way two-party communication setting for instances of \textsc{Part}
with $m = 1$ and $p = 2$. Alice is given a sequence $Y \in \{0, 1\}^n$ and Bob is given a sequence $Z \in \{0, 1 \}^n$,
and they have to split the sequence $X = Y \circ Z$ into two parts. Alice sends a message to Bob, and upon reception,
Bob outputs the separator. We describe now the hard input distribution. 

Let $t$ be an integer that is to be determined later. Alice's input and Bob's input are independent from each other and
they are constructed as follows:

\textbf{Alice's input} $Y$ is a sequence of length $n$ with $2(t-1)$ leading $1$s, followed by an arbitrary sequence 
of length $n - 3t + 2$, with elements from $\{0,11\}$ ($11$ is a pair of ones), where the number of $11$s is exactly $t$. 
Denote by $\mathcal{Y}$ the set of all such sequences. Then $Y$ is chosen uniformly at random from $\mathcal{Y}$.
Clearly, the weight of $Y$ is $4t-2$, and 
$
|\mathcal{Y}| = { n - 3t + 2 \choose t }. \label{eqn:392}
$

%are placed uniformly at random. A valid instance for $t = 2$ and $n = 10$ is for example $11\,00110110$. The total weight of $Y$ is 
%$4t - 2$. There are
%different inputs for Alice which can be seen as follows. The rightmost $n - 2(t-1)$ positions contain $t$ pairs 
%of $1$s ($2t$ ones in total) and $n-4t+2$ $0$s. Think about a pair of ones as a single element. Then the cardinality of the
%input distribution is equal to the number of strings of length $n - 2(t-1) - t$ that contain precicely $t$ pairs
%of ones (seen as a single element) and $n-4t+2$ $0$s which explains the previous bound. 
 \textbf{Bob's input} $Z$ is a sequence of length $n$ with the first $4(i-1)$ elements $1$, and the remaining elements $0$, for some $i \in \{1, 2, \dots, t \}$. Denote all such sequences as $\mathcal{Z}$. Then $Z$ is chosen uniformly at random from $\mathcal{Z}$.
Observe that the weight of $Z$ varies from $0$ to $4(t-1)$, and $|\mathcal{Z}| = t$.

%$\mathcal{Z} = \{ \underbrace{1 \dots 1}_{4(i-1)} \underbrace{0 \dots 0}_{n - 4(i-1)} \, | \, i \in \{1, 2, \dots, t \} \}$.

Note that an optimal partitioning of any $Y \circ Z$ instance splits one of the $11$s in the second part of Alice's input. 

\emph{Example:} Let $t=2$ and $n=10$ and $p=2$. Suppose that Alice holds $Y = 11\,00110110$. Bob's possible inputs are $Z_1 = 0000000000$ and 
$Z_2 = 1111000000$ of weight $0$ and $4$. The optimal partitioning of $Y \circ Z_1$ is $11\, 00 1 \, | \, 10110 \, 0 \dots 0$ and of $Y \circ Z_2$ is $11 \, 001101 \, | \, 1 0 \, 11110 \dots 0$.

We give a lower bound on the space requirement of any possibly randomized communication protocol that solves 
instances of $\mathcal{Y} \times \mathcal{Z}$ exactly.

\vspace{-0.15cm}

\begin{lemma} \label{lem:exact-lb}
 Any randomized one-way two-party communication protocol with error at most $\delta > 0$ that solves \textsc{Part} on instances of
 $\mathcal{Y} \times \mathcal{Z}$ has communication complexity at least
\begin{eqnarray*}
\log \left( \frac{ { n - 3t + 2 \choose t } }{8 {t  \choose 4 \delta t } n^{4 \delta t}} \right). 
\end{eqnarray*}
\end{lemma}
\begin{proof}
 Let $P$ be a randomized protocol as in the statement of the lemma. Then by Yao's Lemma \cite{y77}, there is a deterministic
protocol $Q$ with distributional error at most $\delta$ that has the same communication complexity. We prove a lower bound on the communication complexity of $Q$.

Denote by $M_1, \dots, M_k$ the possible messages from Alice to Bob, and let $\mathcal{Y}_i \subseteq \mathcal{Y}$ 
denote the set of inputs that Alice maps to message $M_i$. Note that for a fixed input for Bob, the protocol $Q$ outputs 
the same result for all inputs in $\mathcal{Y}_i$. We define:
\begin{eqnarray*}
 p_i = \Pr_{Y \gets \mathcal{Y}_i, Z \gets \mathcal{Z} } [ \text{$Q$ errs on $(Y, Z)$} ].
\end{eqnarray*}
Since the distributional error of the protocol is $\delta$, or in other words $\Pr_{Y \gets \mathcal{Y}, Z \gets \mathcal{Z} } [ \text{$Q$ errs on $(Y, Z)$} ] \le \delta$,  
we obtain $\frac{\sum_i p_i |\mathcal{Y}_i|}{|\mathcal{Y}|} \le \delta$.
Let $i \in \{1, \dots, l\}$ be the indices for which $p_i \le 2\delta$. Then by the Markov Inequality,
$\sum_{i=1}^l |\mathcal{Y}_i| \ge \frac{1}{2} |\mathcal{Y}|$.
%Argue now with Markov Inequality that for at least $1/2$ of the inputs the error is at most twice as large (do not chose correct 
%with 2/3 but much higer).

We bound $|\mathcal{Y}_i|$ from above for all $i \in \{1, \dots, l\}$.
%This in turn implies that $l$, the number of messages on which the protocol errs with probability at most $2 \epsilon$, is large.
%if $|\mathcal{Y}_i|$ is large, i.e., many inputs get mapped to the same message, then $p_i$ is small.
First, note that for a particular input $Z \in \mathcal{Z}$, 
the output of $Q$ on $(Y, Z)$ is the same for all $Y \in \mathcal{Y}_i$. Denote by $\mathcal{Y}_i^j$ the subset of $\mathcal{Y}_i$ such 
that for each $Y^j \in \mathcal{Y}_i^j:$ $\Pr_{Z \gets \mathcal{Z}} [ \text{$Q$ errs on $(Y^j, Z)$} ] = \frac{j}{t}$,
or in other words, there are $j$ inputs of Bob such that the protocol fails on $Y^j$, and for the remaining $t-j$ inputs of Bob,
the protocol succeeds. Consider the set $\mathcal{Y}_i^0$, i.e., for each $Y \in \mathcal{Y}_i^0$, the protocol 
succeeds on any input of Bob. This determines all positions of the pairs of $1$s in Alice's input, and therefore,
there is only a single such element and we obtain $|\mathcal{Y}_i^0| \le 1$. Similarly, we obtain:
\begin{eqnarray*}
 |\mathcal{Y}_i^j| \le {t  \choose j } n^j,
\end{eqnarray*}
since the protocol errs on at most $j$ inputs of Bob, therefore the position of $t-j$ pairs of $1$s is fixed and only
$j$ pairs of $1$s may differ (we allow them to have an arbitrary position in $Y$ which is a very rough estimate). 

We apply the Markov Inequality again: for at least half of the elements of $\mathcal{Y}_i$, the protocol errs with 
probability at most $4 \delta$. Therefore:
\begin{eqnarray*}
\frac{1}{2} | \mathcal{Y}_i| & \le & \sum_{j \le 4 \delta t} |\mathcal{Y}_i^j| \le \sum_{j \le 4 \delta t} {t  \choose j } n^j \le 2 {t  \choose 4 \delta t } n^{4 \delta t}, % \le 2 (nt)^{4 \epsilon t},
\end{eqnarray*}
and thus $| \mathcal{Y}_i| \le 4 {t  \choose 4 \delta t } n^{4 \delta t}$. 
This implies that:
\begin{eqnarray*}
 l & \ge & \frac{|\mathcal{Y}|}{8 {t  \choose 4 \delta t } n^{4 \delta t}} = \frac{ { n - 3t + 2 \choose t } }{8 {t  \choose 4 \delta t } n^{4 \delta t}}.
 % \ge \frac{\left( \frac{n - 3t + 2}{t} \right)^t} {8 n^{4 \delta t} \left(\frac{e}{4 \delta}\right)^{4 \delta t} },
\end{eqnarray*} 
%where we used $\left( \frac{x}{y} \right)^y \le  {x \choose y}  \le \left( \frac{e \cdot x}{y} \right)^y$.
Since the protocol sends at least $l$ different messages, the communication complexity of the protocol is at least $\log(l)$, which implies the result.
\qed
% We also have:
% \begin{eqnarray*}
%  p_i & = & \sum_{j} \frac{j}{t} \frac{|\mathcal{Y}_i^j|}{|\mathcal{Y}_i|} \le \sum_{j}^{m} \frac{j}{t} \frac{{t  \choose j } n^j}{|\mathcal{Y}_i|} \\
%  & \le & \sum_{j}^{m-1} \frac{{t  \choose j } n^j}{|\mathcal{Y}_i|} + (1 - \frac{m}{t}) \frac{ { t \choose m} n^m}{|\mathcal{Y}_i|} \\
%  & \le & \frac{1}{2} + \frac{1}{2} ( 1 - \frac{m}{t} ) = 1 - \frac{1}{2} \frac{m}{t},
% \end{eqnarray*}
% where $m$ is the largest value $m'$ such that $2 {t  \choose m' } n^{m'} \le |\mathcal{Y}_i|$.
% In other words, $p_i > 1 - \epsilon$ implies that $\epsilon > \frac{1}{2} \frac{m}{t}$ or $m < 2 t \epsilon$.
% This bounds the size of $\mathcal{Y}_i$ by $2 {t  \choose 2 t \epsilon } n^{2 t \epsilon}$.
\end{proof}

We choose $t$ small enough so that a solution to any instance of $\mathcal{Y} \times \mathcal{Z}$ 
that is a $(1+\epsilon)$-approximation actually solves the instance exactly. This idea leads to our main lower bound theorem:
\begin{theorem} \label{thm:lb-approx}
 Any randomized one-way two-party communication protocol with error at most $\delta > 0$ ($\delta$ sufficiently small) 
that computes a $(1+\epsilon)$-approximation ($\frac{1}{\epsilon} = O(n^{1-\gamma})$ for any $\gamma > 0$) to 
\textsc{Part} on instances of $\mathcal{Y} \times \mathcal{Z}$ has communication complexity at least
%\begin{eqnarray*}
$\Omega \left( \frac{1}{\epsilon} \log n  \right) .$ 
%\end{eqnarray*}
\end{theorem}

\begin{proof}
We choose $t$ small enough that a solution to any instance of $\mathcal{Y} \times \mathcal{Z}$ that is a $(1+\epsilon)$-approximation actually solves the instance exactly. Remark again that the weight of $Y$ is 
$4t-2$ and the weight of $Z$ is $4(i-1)$. Since the total weight is even, there is always a partitioning with 
weight $2t-1+2(i-1)$. Therefore, any partitioning that does not achieve an optimal balancing has an
approximation factor of at least $\frac{2t-1+2(i-1) + 1}{2t-1+2(i-1)}$, and we wish to choose $t$ such that this approximation
factor is worse than a $(1+\epsilon)$ approximation. Therefore, we have to choose $t$ small enough such 
that for any $i \in \{1, 2, \dots, t \}$
\begin{eqnarray*}
 \frac{1}{2t-1+2(i-1)} > \epsilon,
\end{eqnarray*}
which implies that $t < \frac{1}{4 \epsilon} + \frac{3}{4}$. We choose $t = \frac{1}{4 \epsilon}$ and plug
this value into the communication lower bound from Lemma~\ref{lem:exact-lb}. Using standard bounds on binomial coefficients:
\begin{eqnarray*}
\Omega \left( \log \left( \frac{ { n - 3t + 2 \choose t } }{8 {t  \choose 4 \delta t } n^{4 \delta t}} \right) \right) & = & \Omega \left( \log \left( \frac{\left( 4\epsilon(n - \frac{3}{4 \epsilon} + 2) \right)^{\frac{1}{4\epsilon}}} {8 n^{\delta/\epsilon} \left(\frac{e}{4 \delta}\right)^{\delta/\epsilon} }   \right) \right) \\
& = & \Omega \left( \frac{1}{4\epsilon} \log (4\epsilon n - 3 + 8 \epsilon) - \frac{\delta}{\epsilon} \log ( \frac{n e}{4 \delta} ) \right) \\
& = & \Omega \left( \frac{1}{4\epsilon} \log (4\epsilon n) - \frac{\delta}{\epsilon} \log ( \frac{n e}{4 \delta} ) \right) \\
 & = & \Omega( \left( \frac{1}{\epsilon} \log n \right),
\end{eqnarray*}
for a sufficiently small but constant $\delta$, and $\epsilon = \Order( n^{1-\gamma})$ for any $\gamma > 0$. This proves the result. \qed

\end{proof}

\section{Conclusion and Open Problems} \label{sec:open-problems}
In this paper, we presented one-pass $(1+\epsilon)$-approximation streaming algorithms for partitioning integer 
sequences that are based on the parametric search framework. We designed a new method for carrying out 
feasibility tests of multiple parameters simultaneously, leading to an improvement over the na\"{i}ve application 
of the method. We compromised our algorithms with lower bounds showing that an optimal solution cannot be computed 
with sublinear space, and a $(1+\epsilon)$-approximation requires space $\Omega(\frac{1}{\epsilon})$, rendering
our algorithms almost tight with respect to the dependency on parameter $\epsilon$.

We demonstrated that the parametric search framework can successfully be applied in the streaming setting,
and even though the streaming model is very restrictive, it allows an improvement over the na\"{i}ve application
of the method. We believe that other problems admit parametric search algorithms in the streaming setting, and
we leave the identification and the study of those as an open problem. 

The most intriguing open question concerns the situation where a streaming algorithm has no information about
the problem parameters $m$, $n$, and $S$, the maximal weight of an element of the stream, the stream length, and the
total weight of the stream, respectively. For this situation, we designed a $2$-approximation algorithm, 
however, there is no argument contradicting the existence of a $(1+\epsilon)$-approximation algorithm.

\bibliography{partitioning}

\begin{thebibliography}{10}

\bibitem{b88}
Bokhari, S.H.:
\newblock Partitioning problems in parallel, pipeline, and distributed
  computing.
\newblock IEEE Trans. Comput. \textbf{37}(1) (1988)  48--57

\bibitem{hl92}
Hansen, P., Lih, K.W.:
\newblock Improved algorithms for partitioning problems in parallel, pipelined,
  and distributed computing.
\newblock IEEE Trans. Comput. \textbf{41}(6) (1992)  769--771

\bibitem{mo95}
Manne, F., Olstad, B.:
\newblock Efficient partitioning of sequences.
\newblock IEEE Trans. Comput. \textbf{44}(11) (1995)  1322--1326

\bibitem{ms95}
Manne, F., S{\o}revik, T.:
\newblock Optimal partitioning of sequences.
\newblock J. Algorithms \textbf{19}(2) (1995)  235--249

\bibitem{kms97}
Khanna, S., Muthukrishnan, S., Skiena, S.:
\newblock Efficient array partitioning.
\newblock In: Automata, Languages and Programming. Volume 1256., Springer
  Berlin Heidelberg (1997)  616--626

\bibitem{hcn92}
Han, Y., Narahari, B., Choi, H.A.:
\newblock Mapping a chain task to chained processors.
\newblock Inf. Process. Lett. \textbf{44}(3) (1992)  141--148

\bibitem{mp97}
Miguet, S., Pierson, J.M.:
\newblock Heuristics for 1d rectilinear partitioning as a low cost and high
  quality answer to dynamic load balancing.
\newblock In: Proceedings of the International Conference and Exhibition on
  High-Performance Computing and Networking. HPCN Europe '97, London, UK, UK,
  Springer-Verlag (1997)  550--564

\bibitem{pa04}
Pinar, A., Aykanat, C.:
\newblock Fast optimal load balancing algorithms for 1d partitioning.
\newblock J. Parallel Distrib. Comput. \textbf{64}(8) (2004)  974--996

\bibitem{f91}
Frederickson, G.N.:
\newblock Optimal algorithms for tree partitioning.
\newblock In: Proceedings of the Second Annual ACM-SIAM Symposium on Discrete
  Algorithms. SODA '91, Philadelphia, PA, USA, Society for Industrial and
  Applied Mathematics (1991)  168--177

\bibitem{sw05}
Schamberger, S., Wierum, J.M.:
\newblock Partitioning finite element meshes using space-filling curves.
\newblock Future Gener. Comput. Syst. \textbf{21}(5) (2005)  759--766

\bibitem{k11}
Konrad, C.:
\newblock Two-constraint domain decomposition with space filling curves.
\newblock Parallel Comput. \textbf{37}(4-5) (2011)  203--216

\bibitem{b13}
Bader, M.:
\newblock Space-Filling Curves - An Introduction with Applications in
  Scientific Computing. Volume~9 of Texts in Computational Science and
  Engineering.
\newblock Springer-Verlag (2013)

\bibitem{mut05}
Muthukrishnan, S.:
\newblock Data streams: Algorithms and applications.
\newblock In: Foundations and Trends in Theoretical Computer Science.
\newblock Now Publishers Inc (2005)

\bibitem{i91}
Iqbal, M.A.:
\newblock Approximate algorithms for partitioning problems.
\newblock International Journal of Parallel Programming \textbf{20}(5) (1991)
  341--361

\bibitem{m78}
Megiddo, N.:
\newblock Combinatorial optimization with rational objective functions.
\newblock In: Proceedings of the Tenth Annual ACM Symposium on Theory of
  Computing. STOC '78, New York, NY, USA, ACM (1978)  1--12

\bibitem{m83}
Megiddo, N.:
\newblock Applying parallel computation algorithms in the design of serial
  algorithms.
\newblock J. ACM \textbf{30}(4) (1983)  852--865

\bibitem{kn97}
Kushilevitz, E., Nisan, N.:
\newblock Communication complexity.
\newblock Cambridge University Press (1997)

\bibitem{y77}
Yao, A.C.C.:
\newblock Probabilistic computations: Toward a unified measure of complexity.
\newblock In: Proceedings of the 18th Annual Symposium on Foundations of
  Computer Science. FOCS '77, Washington, DC, USA, IEEE Computer Society (1977)
   222--227

\end{thebibliography}

%\newpage

%\input{appendix-alternative.tex}

\end{document}